\documentclass[letterpaper,10pt,conference]{ieeeconf}
\IEEEoverridecommandlockouts
\usepackage{preamble}

\usepackage{tikz}
\usepackage{pgfplots}
\usetikzlibrary{calc,patterns.meta}
\usepgfplotslibrary{fillbetween}
\pgfplotsset{compat=1.17} 

\begin{document}
\title{\LARGE \bf Stochastic Model Predictive Control\\ using Initial State and Variance Interpolation}
\author{Henning Schl\"{u}{}ter and Frank Allg\"{o}{}wer
\thanks{
Manuscript received 31 March 2023;
revised XX XXXX 2023;
accepted XX XXXX 2023.
Date of current version 31 March 2023. 
The authors thank the German Research Foundation (DFG) for financial support under the Grant GRK 2198/1-277536708, and Grant AL 316/12-2-279734922;
and the International Max Planck Research School for Intelligent Systems (IMPRS-IS) for supporting Henning Schl\"{u}{}ter.
(Corresponding author:
Henning Schl\"{u}{}ter.)}
\thanks{The authors are with Institute for Systems Theory and Automatic Control, University of Stuttgart, 70569 Stuttgart, Germany (e-mail: \{\href{mailto:henning.schlueter@ist.uni-stuttgart.de}{schlueter},\href{mailto:frank.allgower@ist.uni-stuttgart.de}{allgower}\}@\href{https://www.ist.uni-stuttgart.de/}{ist.uni-stuttgart.de}).}}

\maketitle

\begin{abstract}
  We present a Stochastic Model Predictive Control (SMPC) framework for linear systems subject to Gaussian disturbances. 
  In order to avoid feasibility issues, we employ a recent initialization strategy, optimizing over an interpolation of the initial state between the current measurement and previous prediction.
  By also considering the variance in the interpolation, we can employ variable-size tubes, to ensure constraint satisfaction in closed-loop.
  We show that this novel method improves control performance and enables following the constraint closer, then previous methods.
  Using a DC-DC converter as numerical example we illustrated the improvement over previous methods.
\end{abstract}


\section{Introduction}
Many problems in control require achieving good performance, while guaranteeing safety constraints in the face of model uncertainty.
A prominent method to address these requirements is model predictive control (MPC) \cite{Rakovic2019mpchandbook}.
This optimization-based control strategy allows to directly consider constraint and performance \cite{Rawlings2017book}.
These methods can be extended to handle various types of uncertainty \cite{Kouvaritakis2016book}.
In particular, we focus on stochastic MPC, which employs probabilistic information to reduce conservatism, by allowing small probability of constraint violation \cite{Mesbah2016smpcSurvey}.
In this work, we consider an initialization strategy for SMPC, loosely inspired by \cite{Mayne2005rmpcIC} and building upon ideas from \cite{ist:schluter2022a}.

\subsubsection*{Related work}
There are two primary philosophies when it comes to handle uncertainty in MPC. 

On the one hand there are approaches following Chisci~et.\@\,al.\@\ \cite{Chisci2001rmpc}, where the MPC problem is initialized using the measurement directly.
This then requires to consider the new starting point for recursive feasibility, which is generally challenging in the stochastic setting.
While for disturbance with bounded support, such approaches for SMPC \cite{Cannon2011stochstictube,Lorenzen2017guaranteedFeasibility} exist.
These methods are often conservative, due to resorting to discarding parts of the stochastic information for feasibility.

In SMPC for unbounded disturbances, two line of thought are followed to achieve recursive feasibility despite the unboundness.
One is to employ a backup strategy, when initializing with the measurement is infeasible. For example, in \cite{Farina2013fallback,Hewing2018prs}, the last prediction is used for which feasibility can be ensured.
This is taken to the extreme in \cite{Hewing2020indirectSMPC}, where the prediction for the constraint is always initialized with the previous predication.
Feedback can then be achieved by using the actual measurement for the cost function.
Both of these ideas lack feedback on the constraint, which maybe undesirable.

On the other hand there are the schemes based on Mayne~et.\@\,al.\@\ \cite{Mayne2005rmpcIC}.
Here, the predication is initialized from some set containing the measurement and the last predication with the optimizer choosing the optimal starting point.
While this make recursive feasibility a non-issue, ensuring closed-loop constraint satisfaction is the challenge.
SMPC approaches with guarantees following this philosophy are a quite recent development.

First steps in this direction in \cite{ist:schluter2022a,Koehler2022interpolSMPC} initialize the predication by interpolating between the measurement and the last predication. 
Thus, mixing ideas from approaches such as \cite{Farina2013fallback,Hewing2018prs,Hewing2020indirectSMPC}.
The concept have since also applied in other settings, such as data-driven MPC \cite{Mark2023dddrsmpcic} and output-feedback SMPC \cite{Pan2022ddOutputSmpcIC}.
While restrictive only optimizing over a line and employing a fixed-size tube for constraint tightening, the result are quite promising.
We expand upon these ideas by also considering the variance, which enables us to employ variable-sized tubes.
The ability to adapt the tube size allows to closely approach the constraint yielding significant performance gains.

Another important aspect in SMPC is the choice of constraint tightening and propagation method for the stochastic uncertainty \cite{Mesbah2016smpcSurvey}.
In the literature there several solutions available. 
The propagation can be, for example, achieved via randomized methods using scenarios, analytical reformulations, concentration inequalities, or polynomial chaos expansion.
Since we have a known distribution for the disturbance, an exact analytical reformulation is available.

\subsubsection*{Contribution}
In this article, we improve upon the initialization strategy employed by state-of-the-art SMPC framework, primarily \cite{ist:schluter2022a}, but also \cite{Farina2013fallback,Hewing2018prs,Koehler2022interpolSMPC}.
We study linear systems with additive Gaussian disturbances subject to individual chance constraints on state and input.
We relax the fixed-size tubes use in \cite{ist:schluter2022a}.
Considering the error variance under the initialization yields variable-sized tubes.
We show that the resulting scheme still satisfies the chance constraint in closed-loop, while achieving a superior performance.
Beyond that in a numerical study we find evidence that the assumptions for the reconditioning of the chance constraints can be relaxed beyond the proven limits.
This indicates that the interpolation improves the chance constraint interpretation, even if the measurement is never chosen exactly as is required by the theory.

\subsubsection*{Notation} 
We use $x_{i}(k)$ to denote $x$ predicted $i$ time steps ahead from time $k$ and $x(k)$ for quantities realized in close-loop.
The probability of $A$, the expected value of $x$ and the variance of $x$ are written as $\Prob(A)$, $\Expectation[x]$, and $\Variance[x]$, respectively.
The weighted 2-norm is denoted as $\xVert x\xVert^{2}_Q = x^{\smash{\T}}\!Qx$ for positive-definite matrices $Q\succ{}0$ and a sequence as $x_{0:N}\coloneqq{}\{x_0, \dots, x_N\}$.
$\dlyap(A,Q)\coloneqq{}P$ is the positive-definite solution of the Lyapunov equation $APA^{\smash{\T}}+Q=P$. 
The columnwise vectorization of a matrix $A$ is denoted by $\vectorize\{A\}$.

\section{Problem Setup} 

Consider the linear discrete-time system
\begin{equation}\label{eq:sys}
  x(t+1) = A x(t) + B u(t) + w(t)
\end{equation}
with additive Gaussian noise $w(t)\sim{}\mathcal{N}{}(0,\mathbb{V}\mkern-2mu{}w)$, state $x(t)\in{}\mathbb{R}{}^{n_x}$, and input $u(t)\in{}\mathbb{R}{}^{n_u}$.

The goal is to design a state feedback MPC minimizing expected LQR cost subject to individual chance constraints
\vspace*{-\baselineskip}
\begin{subequations}\label{eq:cc}
\begin{alignat}{5}
  \Prob(c_{i}^{\smash{\T}}x_{k}(t) \leq{} d_{i}) &\geq{} \varrho{}_{i}  \qquad&&\forall{}\;\;0&&\leq{}i&&<n_{cx}\\
  \Prob(c_{j}^{\smash{\T}}u_{k}(t) \leq{} d_{j}) &\geq{} \varrho{}_{j}  \qquad&&\forall{}n_{cx}&&\leq{}i&&<n_{cx}+n_{cu}
\end{alignat}
\end{subequations}
on state and input with $c_{i}\in{}\mathbb{R}{}^{n_{x}}$, $c_{j}\in{}\mathbb{R}{}^{n_{u}}$, and $d_{i},d_{j}\in{}\mathbb{R}{}$.
The conditioning of these probabilities is specified later in \cref{sec:constraint tightening}.
The guaranteed satisfaction probability $\varrho{}\in{}(0,1)$ allows for some acceptable risk of constraint violation.
Due to the unbounded support of the Gaussian noise, hard constraints cannot be enforced.

\begin{remark}
  While beyond the scope of this work, the algorithm can be adapted to handle joint chance constraint using one of the inequalities in \cite{Hochberg1987multiplecomparison}, \eg Boole's inequality. Thereby the joint constraint can be split into individual constraints, while assigning their violation risk based on these inequalities online.
\end{remark}

As is common in the literature \cite{Mesbah2016smpcSurvey}, we employ a disturbance-affine feedback parameterization. To this end the system is prestabilized by the matching LQR controller $K$.
This allows us to split the dynamics into a nominal and an error part, which is purely driven by the disturbance.
Thus, choosing $x=z+e$, we obtain
\begin{align*}
  z_{k+1}&=Az_{k}+Bv_{k}\\
  e_{k+1}&=(A+BK)e_{k}+w_{k}
\end{align*}
with the input $u=v+Ke$. 
Thereby, the nominal part is a deterministic process, while the error is a Gaussian random process, since only it is driven by the Gaussian noise.
Thus, we can compute the variance of the error
\begin{align*}
  \mathbb{V}\mkern-2mu{}e_{k+1}&=(A+BK)\mathbb{V}\mkern-2mu{}e_{k}(A+BK)^{\smash{\T}}+\mathbb{V}\mkern-2mu{}w\,,
\end{align*}
which yields in absence of the constraints the steady-state variance $\mathbb{V}\mkern-2mu{}e_\infty{}\coloneqq{}\dlyap(A+BK,\mathbb{V}\mkern-2mu{}w)$.

\section{Stochastic Model Predictive Control} 

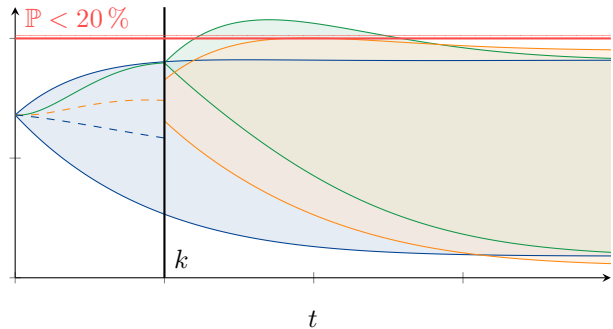
\begin{figure}
  \centering
  \bgroup
    \definecolor{istblue}{rgb/cmyk}{0.0,0.25490,0.56862745/1,0.55,0.0,0.43}
    \definecolor{istgreen}{rgb/cmyk}{0.388235,0.83137,0.4431372/0.53,0,0.46,0.16}
    \definecolor{istorange}{rgb/cmyk}{1.0,0.5333,0.0667/0.0,0.46,0.93,0.0}
    \definecolor{istred}{rgb/cmyk}{0.9961,0.2902,0.2863/0.0,0.7,0.71,0.0}
    \definecolor{istlightblue}{rgb/cmyk}{0.3765, 0.6863, 1.0/0.62,0.31,0.0}
    \definecolor{istdarkblue}{rgb/cmyk}{0.1176,0.1804,0.8706/0.86,0.79,0.0,0.12}
    \definecolor{istdarkgreen}{rgb/cmyk}{0.0627,0.5882,0.2824/0.89,0.0,0.52,0.41}
    \definecolor{istdarkred}{rgb/cmyk}{0.529,0.031,0.075/0,0.94,0.86,0.47}
    \begin{tikzpicture}
        \begin{axis}[{
            declare function={
                tube(\x)=0.91*(\x>0)*(1-exp(-1.5*(\x)));
                trajA(\x)=0.4*(1+sin(45*(1.4*\x)))*exp(-1.4*(\x));
                trajB(\x)=0.7*(\x>0)*(1+sin(45*(1.4*\x-2)))*exp(-(1.4*\x-2));
                trajC(\x)=0.5*trajB(\x);
                scaleC=1.1;
                offC=0.86;
            },
            axis lines=left,
            grid=none,
            samples=400,
            xmin=0, xmax=3.99, ymin=-1.11, ymax=1.4,
            xtick distance=1,
            ytick={-1.11,0,1.11},
            height=0.6\linewidth,width=1.1\linewidth,
            ylabel={},yticklabels={},
            xlabel={$t$},xticklabels={},
        }]
            \addplot[name path=Au, smooth, domain=0:4, istblue] {trajA(x)+tube(x)};
            \addplot[name path=Al, smooth, domain=0:4, istblue] {trajA(x)-tube(x)};
            \addplot[istblue!20, opacity=0.5] fill between [of=Au and Al];

            \addplot[name path=Bm, smooth, domain=0:1, istdarkgreen] {trajA(x)+trajB(x)};

            \addplot[name path=Bu, smooth, domain=1:4, istdarkgreen] {trajA(x)+trajB(x)+tube(x-1)};
            \addplot[name path=Bl, smooth, domain=1:4, istdarkgreen] {trajA(x)+trajB(x)-tube(x-1)};
            \addplot[istdarkgreen!20, opacity=0.5] fill between [of=Bu and Bl];

            \addplot[name path=Cm, smooth, domain=0:1, istorange, dashed] {trajA(x)+trajC(x)};
            \addplot[name path=Am, smooth, domain=0:1, istblue, dashed] {trajA(x)};
            \addplot[name path=Cu, smooth, domain=1:4, istorange] {trajA(x)+trajC(x)+scaleC*tube(x-offC)};
            \addplot[name path=Cl, smooth, domain=1:4, istorange] {trajA(x)+trajC(x)-scaleC*tube(x-offC)};
            \addplot[istorange!20, opacity=0.5] fill between [of=Cu and Cl];
            \draw[black,thick] (1,\pgfkeysvalueof{/pgfplots/ymin}) -- node[pos=0.0,anchor=south west] {$k$} (1,\pgfkeysvalueof{/pgfplots/ymax});
            \begin{scope}[pattern color=istred,istred]
                \pattern[pattern={Lines[angle=45,distance={2pt/sqrt(2)},line width=0.25pt]}] (\pgfkeysvalueof{/pgfplots/xmin},1.14) rectangle (\pgfkeysvalueof{/pgfplots/xmax},1.11);
                \draw[thick] (\pgfkeysvalueof{/pgfplots/xmin},1.11) -- node[pos=0.0,anchor=south west] {$\mathbb{P}<20\,\%$} (\pgfkeysvalueof{/pgfplots/xmax},1.11);
            \end{scope}
        \end{axis}
    \end{tikzpicture}
\egroup
  \caption{
    Illustrates how conditioning the probability distribution on different information affect the interpretation of a chance constraint.
    Shown in blue the distribution without taking measurement information into account, \ie conditioned on the initial time.
    In green, the distribution conditioned on the current measurement at time $k$ is shown, which due to the disturbance up to $k$ violates the constraint conditioned on time $k$.
    Lastly, in orange we show our proposed solution considering the interpolation between these extreme. 
    The shown confidence bounds of the distribution match the chance constraint, \ie should lie below the constraint indicated in red.
  }
  \label{fig:reconditioning}
\end{figure}

One of the central problems in formulating the SMPC problem is achieving recursive feasibility in spite of the chance constraints permitting violations up to a probability. 
In particular, this problem arises when taking new measurement into account as illustrated in \cref{fig:reconditioning}.
The predicted distribution of the future system state is based on all possible future trajectories for all possible disturbance realizations. 
By taking the measurement into account, we select a subset of these trajectories.
Thereby the percentage of trajectories that violate the constraint may increase beyond the proscribed limit, as illustrated by the confidence bounds in the figure.

To overcome this obstacle, there are three established method in the literature.
While, for bounded disturbances, one often conservatively relies on methods from robust MPC to ensure recursive feasibility, for unbounded disturbances, commonly the fact that the problem become infeasible with a non-zero probability is permitted. 
Then, either the recursive feasibility becomes probabilistic or a fallback strategy \cite{Farina2013fallback,Hewing2018prs} is devised for when the original problem is inevitably infeasible.
Lastly, some approach even simply argue against taking the measurement into account for the chance constraint \cite{Hewing2020indirectSMPC}.

Recently, we proposed a novel approach \cite{ist:schluter2022a} to overcome this issue by interpolating between the measurement and the last predicted nominal state, thus initializing with
\begin{equation}\label{eq:interpolation}
  z_{0}(t) = (1-\xi{}(t)) x(t) + \xi{}(t) z_{1}(t{-}1)\,.
\end{equation}
Thus far, that approach was not able to take the error distribution into account, thus the tightening of the chance constraint has to rely conservatively on some fixed upper bound on the distribution.
Hence, that tightening is the same regardless of how close to the true measurement the prediction was initialized.

The work rectify this limitation by interpolating the entire distribution, as visualized in orange in \cref{fig:reconditioning}.
For simplicity, we assume the measurement is exact, \ie has zero variance.
Then, since the state distribution is Gaussian due to the Gaussian noise, we only need to consider the expected value and variance for the interpolation.
Thus, for the variance we obtain
\begin{equation}\label{eq:interpolation:var}
  \mathbb{V}\mkern-2mu{}e_{0}(t) = \Expectation[(x(t)-z_{0}(t))^{2}] = \xi{}^{2}(t) \mathbb{V}\mkern-2mu{}e_{1}(t{-}1)\,,
\end{equation}
while the expected value is the nominal state, hence, the same as \eqref{eq:interpolation}.
This type of can be interpreted as conditioning on a prior artificial point in time.
In another interpretation we compare this to an observer such a Kalman filter, where the \emph{a posteriori} state estimate is derived from a mixture of the predicted and measured state. 
While this interpolation is vastly more simplistic, it is based in the same idea that the measurement is not necessary is the correct or best state to work with.
In particular, we are not required to use the true or expected state to initialize the optimal control problem for the SMPC scheme.
As long as we can show that the scheme ensure the chance-constraint in closed-loop, the interpretation for the open-loop prediction is secondary.
We will show that initializing with this interpolation, yields is meaningful for the closed-loop constraint via \cref{thm:cc:closedloop}, and yields a meaningful cost function.

In the remainder of this section, we derive a deterministic surrogate for the SMPC problem by constructing a constraint tightening, a cost function, and terminal ingredients.

\subsection{Constraint Tightening} \label{sec:constraint tightening}
In order to tighten the chance constraint, we first assume that they are conditioned such that $\mathbb{E}{}e=0$.
Then, employing that the noise is Gaussian by assumption, the state chance constraints can be restated as
\begin{gather}
  \varrho{}_{i}\leq{}\Prob(c_{i}^{\smash{\T}}x\leq{}d_{i})=\Prob(c_{i}^{\smash{\T}}e\leq{}d-c_{i}^{\smash{\T}}z)=\Phi{}\!\left(\frac{d_{i}-c_{i}^{\smash{\T}}z}{\sqrt{c^{\smash{\T}}_{\smash{\raisebox{-.3ex}{$\scriptstyle i$}}} \mathbb{V}\mkern-2mu{}e c_{i}\,}}\right)\notag\\
  \mathllap{\Longleftrightarrow{}\quad}c_{i}^{\smash{\T}}z\leq{}d_{i}-\sqrt{c^{\smash{\T}}_{\smash{\raisebox{-.3ex}{$\scriptstyle i$}}} \mathbb{V}\mkern-2mu{}e c_{i}\,} \Phi{}^{-1}(\varrho{}_{i}) \label{eq:constraint:tightened:state}
\end{gather}
with the cumulative density function of the standard normal distribution $\Phi{}(x)$.
Similarly, one can derive
\begin{equation}
  c_{j}^{\smash{\T}}v\leq{}d_{j}-\sqrt{c^{\smash{\T}}_{\smash{\raisebox{-.3ex}{$\scriptstyle\! j$}}} K\mathbb{V}\mkern-2mu{}e K^{\smash{\T}}c_{\smash{\! j}}\,} \Phi{}^{-1}(\varrho{}_{j}) \label{eq:constraint:tightened:input}
\end{equation}
as deterministic replacement for input chance constraints.
\begin{remark}
This constraint tightening is exact, but just as the chance constraints themselves the resultant constraint is not convex, unless $\varrho{}=0.5$, in which case the problem would reduce to expectation constraints.
Thus, one uses either a nonlinear solver or employs a convex approximation of the tightened constraint.
Two common methods to approximate these constraints \cite{Lucia2015pceconvexification} are either a first-order Taylor approximation or split the constraint into one on the variance and one on the nominal value by fixing the value of the square root.
Either approximation is update based on the last solution. 
\end{remark}

\subsection{Cost Function}\label{sec:smpc:cost}
Since the objective is to minimize the expected LQR cost over the infinite horizon, we first have to ensure that the cost remains finite despite the additive noise in the system. Therefore, we use the cost function
\begin{equation}\label{eq:cost:original}
  J(t)=\mathbb{E}{}_{t}\!\left[\,\sum_{k=0}^\infty{}\ns{}\Bigl(\ns{} \Vert{}x_{k}\Vert{}^{2}_Q + \Vert{}u_{k}\Vert{}^{2}_R - \ell{}_{ss}\ns{}\Bigr)\right]
\end{equation}
as in \cite[\abvSec6.2]{Kouvaritakis2016book} with the expected steady-state step-cost $\ell{}_{ss}\coloneqq{}\trace[\p(Q{+}K^{\smash{\T}}RK)\mathbb{V}\mkern-2mu{}e_\infty{}]$ in absence of the constraints, where $\mathbb{E}{}_{t}$ is the exception conditioned on the measurement at the current time $t$.
With this conditioning we do not have that the expected error $\mathbb{E}{}_{t}e$ is zero, thus reformulating with the decomposition $x=z+e$ yields
\begin{align*}
  &=\!\! \sum_{k=0}^\infty{}\ns{}\Bigl(\ns{} \Vert{}z_{k}\Vert{}^{2}_Q + \Vert{}v_{k}\Vert{}^{2}_R + \trace[\p(Q{+}K^{\smash{\T}}RK)\mathbb{E}{}_{t}\left[ee^{\smash{\T}}\right]] - \ell{}_{ss} \ns{}\Bigr)\\
  &=\!\! \sum_{k=0}^{N-1}\ns{}\Bigl(\ns{} \Vert{}z_{k}\Vert{}^{2}_Q {+} \Vert{}v_{k}\Vert{}^{2}_R {+} \trace[\p(Q{+}K^{\smash{\T}}RK)\mathbb{V}\mkern-2mu{}e_{k}]\ns{}\Bigr)\ns{} +\!\! \underbrace{\sum_{k=0}^\infty{} \Vert{}\mathbb{E}{}_{t}e_{k}\Vert{}^{2}_{Q{+}K^{\smash{\T}}RK}}_{\smash{\eqqcolon{}\,\beta{}_{t}(\cdot{})}}\notag\\
  &\mkern50mu+ \overbrace{\Vert{}z_N\Vert{}^{2}_P +\!\! \sum_{k=N}^\infty{} \ns{}\Bigl(\ns{} \trace[\p(Q{+}K^{\smash{\T}}RK)\mathbb{V}\mkern-2mu{}e_{k}] - \ell{}_{ss}\ns{}\Bigr) - N\ell{}_{ss}}^{\smash{\eqqcolon{}\,\digamma{}(\cdot{})\,+\,\text{const.}}}
\end{align*}
with $P\coloneqq{}\dlyap(A+BK,Q{+}K^{\smash{\T}}RK)$.
The initial cost term $\beta{}$ and terminal cost term $\digamma{}$ are discussed in the following.

\subsubsection{Initial Cost \texorpdfstring{$\mathit{\beta{}_{t}(\xi{}(t)\ns{}\ns{})}$}{}} 
The expected value of the initial error $\mathbb{E}{}_{t}e_{0}(t)$ is not zero due to the conditioning on the current time step, but has to be derived from the initialization via $\mathbb{E}{}_{t}e_{0}(t)=x(t)-z_{0}(t)$. 
Thus, we obtain
\begin{equation*}
\beta{}_{t}(\xi{}(t)\ns{}\ns{}) = \Vert{}\mathbb{E}{}_{t}e_{0}(t)\Vert{}_P = \Vert{}x(t)-z_{1}(t{-}1)\Vert{}^{2}_P \xi{}^{2}(t)
\end{equation*}
which only depends on the interpolation variable $\xi{}$. 

\subsubsection{Terminal Cost \texorpdfstring{$\mathit{\digamma{}(z_N(t),\Sigma{}_N(t)\ns{}\ns{})}$}{}} 
In order to obtain a traceable formulation for the terminal cost $\digamma{}(\cdot{})$, we have to address the infinite sum inside
\begin{align*}
  \digamma{}(\cdot{}) &= \Vert{}z_N\Vert{}^{2}_P - \trace[\p(Q{+}K^{\smash{\T}}RK) \!\!\vphantom{\sum_{k=N}^\infty{}}\smash{\underbrace{\sum_{k=N}^\infty{} \Bigl( \mathbb{V}\mkern-2mu{}e_\infty{} - \mathbb{V}\mkern-2mu{}e_{k}\Bigr)}_{\eqqcolon{}\,\mathcal{S}{}_N}}] + \text{const.}\\[-0.5\baselineskip]
\end{align*}
Using that $\mathbb{V}\mkern-2mu{}w=\mathbb{V}\mkern-2mu{}e_\infty{}-(A+BK)\mathbb{V}\mkern-2mu{}e_\infty{}(A+BK)^{\smash{\T}}$ per definition of $\mathbb{V}\mkern-2mu{}e_\infty{}$, we can then expand the sum 
\begin{align*}
  \mathcal{S}{}_N &= \mathbb{V}\mkern-2mu{}e_\infty{}{-}\mathbb{V}\mkern-2mu{}e_N + \!\!\sum_{\crampedclap{k=N+1}}^\infty{} \bigl( \mathbb{V}\mkern-2mu{}e_\infty{}\! - (A{+}BK) \mathbb{V}\mkern-2mu{}e_{k-1} (A{+}BK)^{\smash{\T}} - \mathbb{V}\mkern-2mu{}w\bigr)\\
  &= \mathbb{V}\mkern-2mu{}e_\infty{}{-}\mathbb{V}\mkern-2mu{}e_N + (A{+}BK) \mathcal{S}{}_N (A{+}BK)^{\smash{\T}}\\
  &= \dlyap\bigl(A+BK, \mathbb{V}\mkern-2mu{}e_\infty{}\bigr) - \dlyap\bigl(A+BK, \mathbb{V}\mkern-2mu{}e_N(t)\bigr)\,.
\end{align*}
Defining $\Sigma{}_N\coloneqq{}\dlyap\bigl(A+BK,\mathbb{V}\mkern-2mu{}e_N(t)\bigr)$, then allows to rewrite $\digamma{}(\cdot{})$ in terms of the optimization variable as
\begin{align*}
  \digamma{}(z_N(t),\Sigma{}_N(t)\ns{}\ns{}) &\coloneqq{} \Vert{}z_N(t)\Vert{}^{2}_P + \trace[\p(Q+K^{\smash{\T}}RK) \Sigma{}_N(t)]
\end{align*}
while dropping the terms that constant in the optimization.
The term $\Sigma{}_N(t)$ is either compute by adding the constraint
\begin{equation*}
  \vectorize\{\mathbb{V}\mkern-2mu{}e_N(t)\} = (\mathbb{I}{}_{n_{x}^{2}} - (A+BK) \otimes{} (A+BK))\vectorize\{\Sigma{}_N(t)\}
\end{equation*}
to compute the required solution of Lyapunov via the optimization.
Alternatively, we can use forward- and reverse-mode derivatives of the Lyapunov equation \cite{Kao2020diffLyapunov} to efficiently use the $\dlyap$-function in the optimization directly.

Thus, we obtain the final cost function for the MPC scheme \eqref{eq:smpc:cost}, which is positive, convex, and differs only by a constant from the original expected LQR cost.

\subsection{Terminal Constraints}
Lastly, in order to ensure recursive feasibility, we employ terminal constraints with the prestabilizing controller as terminal controller.
For the nominal prediction $z_N$ the usual requirements for a terminal set $\mathcal{Z}{}_F$ are sufficient.
That is the terminal set is invariant under the controller, \ie $(A + BK)\mathcal{Z}{}_F \subseteq{} \mathcal{Z}{}_F$ and satisfies the tightened constraints \cref{eq:constraint:tightened:state,eq:constraint:tightened:state,eq:constraint:tightened:input}, such a set can be computed via \cite[Alg.\,2.1]{Kouvaritakis2016book}.

For the error variance we require that $\mathbb{V}\mkern-2mu{}e_N\preccurlyeq{}\mathbb{V}\mkern-2mu{}e_\infty{}$.
This, however, follow from the initialization as $\mathbb{V}\mkern-2mu{}e_{0}(t) \preccurlyeq{} \mathbb{V}\mkern-2mu{}e_{1}(t-1) \preccurlyeq{} \mathbb{V}\mkern-2mu{}e_\infty{}$ and $\mathbb{V}\mkern-2mu{}e_{0}(0) \preccurlyeq{} \mathbb{V}\mkern-2mu{}e_\infty{}$, since $A+BK$ is Hurwitz and $K$ is not optimized. 
Thus, we do not require a terminal constraint for the variance.

\subsection{SMPC Scheme}
Bringing the ideas from the previous sections together, we obtain the deterministic surrogate optimal control problem%
\begin{samepage}%
\begin{subequations}\label{eq:smpc}\bgroup%
  \begin{alignat}{4}
    \min_{\mathclap{\crampedsubstack{v_{k}(t), z_{k}(t),\\ \mathbb{V}\mkern-2mu{}e_{k}(t), \xi{}(t)}}}&\quad &{\quad                               
                            \sum_{k=0}^{N-1} \Vert{}z_{k}(t)\Vert{}^{2}_Q + \Vert{}v_{k}(t)\Vert{}^{2}_R + \trace[(Q+K^{\smash{\T}}RK)\mathbb{V}\mkern-2mu{}e_{k}(t)]                       
    } \span\notag                                                                                                \\[-0.85\baselineskip]
    &&{\qquad                     + \beta{}_{t}(\xi{}(t)) + \digamma{}\bigl(z_N(t),\Sigma{}_N(t)\bigr)}\span             \label{eq:smpc:cost} \\[0.25\baselineskip]
    \mathrm{s.t.} &&          z_{k+1}(t) &= A z_{k}(t) + B v_{k}(t)\,,                                                    \\
    \span\span               \mathbb{V}\mkern-2mu{}e_{k+1}(t) &= (A+BK) \mathbb{V}\mkern-2mu{}e_{k}(t) (A+BK)^{\smash{\T}} + \mathbb{V}\mkern-2mu{}w\,,                                           \\
    \span\span                  z_{0}(t) &= (1-\xi{}(t)) x(t) + \xi{}(t) z_{1}(t{-}1)\,,                                      \\
    \span\span                 \mathbb{V}\mkern-2mu{}e_{0}(t) &= \xi{}^{2}(t) \mathbb{V}\mkern-2mu{}e_{1}(t{-}1)\,,                                \label{eq:smpc:ic:V}\\
    \span\span               c_{i}^{\smash{\T}}z_{k}(t) &\leq{} d_{i} - \sqrt{c_{i}^{\smash{\T}}\mathbb{V}\mkern-2mu{}e_{k}(t)\,c_{i}}\Phi{}^{-1}(\varrho{}_{i}) \,,                      \label{eq:smpc:cc:x} \\
    \span\span               c_{j}^{\smash{\T}}v_{k}(t) &\leq{} d_{j} - \sqrt{c_{j}^{\smash{\T}}K\,\mathbb{V}\mkern-2mu{}e_{k}(t)\,K^{\smash{\T}}c_{j}}\Phi{}^{-1}(\varrho{}_{j}) \,,                 \label{eq:smpc:cc:u} \\
    \span\span                 z_N(t) &\in{} \mathcal{Z}{}_F\,,                                                                  \\
    \span\span  \vectorize\{\mathbb{V}\mkern-2mu{}e_N(t)\} &= (\mathbb{I}{} - (A{+}BK) \otimes{} (A{+}BK))\vectorize\{\Sigma{}_N(t)\}                           \\
    \span\span                   \xi{}(t) &\in{} \left[0,1\right]\,,                                                     \\
    \span\span                     \forall{}k &\in{} \left\{0,\ldots{},N-1\right\}                                           \notag \\
    \span\span                \forall{}i,j \in{}\mathbb{N}{} &: 0\leq{} i < n_{\mathrm{cx}} \leq{} j < n_{\mathrm{cu}} + n_{\mathrm{cx}}      \notag 
  \end{alignat}\egroup%
\end{subequations}%
\end{samepage}%
with input sequence $v_{0:N-1}$, nominal state prediction $z_{0:N}$, error variance $\mathbb{V}\mkern-2mu{}e_{0:N}$, and interpolation variable $\xi{}$ as optimization variables. This problem is solved at each time step yielding the applied input $u(t)=v_{0}(t)+K(x(t)-z_{0}(t))$.

\section{Theoretical Guarantees} 

The foundation of the theoretical results developed in this section is that \cref{eq:smpc} is recursively feasible, which is immediate by construction.

\begin{theorem}\label{thm:recfea}
  If the SMPC scheme \cref{eq:smpc} is feasible at $t=0$ with $z_{1}(-1)=x(0)$ and $\mathbb{V}\mkern-2mu{}e_{1}(t{-}1)\preccurlyeq{}\mathbb{V}\mkern-2mu{}e_\infty{}$, then it is recursively feasible.
\end{theorem}
\begin{proof}
  The shifted previous solution $v_{k}(t)=v_{k+1}(t{-}1)$ is feasible with $\xi{}(t)=1$, since $\mathcal{Z}{}_F$ is invariant and satisfies the tightened constraints.
\end{proof}

While other initialization for $z_{1}(-1)$ might equally yield recursive feasibility possibly even with a larger feasible region, starting with $z_{0}(0)=x(0)$ or $\xi{}(0)=0$, which this choice ensures, is desirable. 
This ensures that there is a point in time to fall back on for the conditioning of the chance constraints on the closed-loop, see \cref{thm:cc:closedloop}.

\subsection{Chance Constraint Satisfaction in Closed-Loop}
Since the constraints tightening assume $\mathbb{E}{}e=0$, \ie $\xi{}=0$, recursive feasibility alone is insufficient to ensure that the chance constraint \cref{eq:cc} are satisfied. 
Thus, additional steps must be taken to ensure that in closed-loop constraints still hold despite not deliberately not always choosing $\mathbb{E}{}e=0$.

In order to address this, we consider the error $e(t)$ for time $t$, where this not necessarily holds true. 
Since the constraint tightening is based upon the error dynamics, this will enable us to consider satisfaction of the chance constraints in closed-loop.
Hence, considering an arbitrary halfspace constraint we obtain the following lemma.

\begin{lemma}\label{lem:cc:closedloop}
  For system \cref{eq:sys} controlled by \cref{eq:smpc} with tightening \cref{eq:constraint:tightened:state,eq:constraint:tightened:state,eq:constraint:tightened:input}, we have for any $r\in{}\mathbb{R}{}$ and any $c\in{}\mathbb{R}{}^{n_{x}}$ that
  \begin{equation*}
    \Prob(\left\vert{}c^{\smash{\T}}e(t)\right\vert{} \leq{} r) \geq{} \Prob(\left\vert{}c^{\smash{\T}}e_{t-t_{0}}(t_{0})\right\vert{} \leq{} r)
  \end{equation*}
  conditioned on $e(t_{0}) = 0$ holds for all $t \geq{} t_{0}$.
\end{lemma}
\begin{proof}
  The predicted error $e_{i}(t)$ depends on the closed-loop disturbances $w(t_{0}),\ldots{},w(t-1)$ and the future disturbances $w_{0}(t),\ldots{},w_{i-1}(t)$ via prediction.
  By collecting the future disturbances in $\tilde{e}_{n}(t) = \sum_{i=0}^{n-1} A_{k}^{n-i-1}w_{i}(t)$, the error becomes $e_{n}(t)=A_{k}^{n} e(t) + \tilde{e}_{n}(t)$.

  Now, in a first step, we show that 
  \begin{equation}
    \Prob(c^{\smash{\T}}e_{n}(t-n)\leq{}r) \srel{!}{\geq{}} \Prob(c^{\smash{\T}}e_{n+1}(t-n-1)) \label{eq:proof::lem:cc:closedloop::step}
  \end{equation}
  holds. Since $e$ is for all time normally distributed around zero, we can use the CDF $\Phi{}$ for the standard normal distribution
  to translate this inequality
  \begin{equation*}
    \Phi{}\!\left(\frac{r}{c^{\smash{\T}}\mathbb{V}\mkern-2mu{}e_{n}(t-n)c}\right) \srel{!}{\geq{}} \Phi{}\!\left(\frac{r}{c^{\smash{\T}}\mathbb{V}\mkern-2mu{}e_{n+1}(t-n-1)c}\right)
  \end{equation*}
  into an inequality on the variance
  \begin{equation*}
    c^{\smash{\T}}\mathbb{V}\mkern-2mu{}e_{n}c \srel{!}{\leq{}} c^{\smash{\T}}\mathbb{V}\mkern-2mu{}e_{n+1}(t-n-1)c
  \end{equation*}
  by exploiting the monotonicity of $\Phi{}$.
  
  Then, by substituting the error we obtain
  \begin{align*}
    &\! c^{\smash{\T}}\mathbb{V}\mkern-2mu{}e_{n}c 
     = c^{\smash{\T}}\mathopen{}\left(A_{k}^{n} \mathbb{V}\mkern-2mu{}e_{0}(t-n) A_{k}^{\smash{\T}}{}^{n} + \mathbb{V}\mkern-2mu{}\tilde{e}_{n}(t-n)\ns{}\right)\mathclose{}c\\
    &\!= c^{\smash{\T}}\ns{}\ns{}\mathopen{}\left(\ns{}\xi{}^{2}\ns{}\ns{}A_{k}^{n+1} \mathbb{V}\mkern-2mu{}e_{0}\ns{}(\ns{}t{\text{\textendash}}n{\text{\textendash}}1\ns{}) A_{k}^{\smash{\T}}{}^{n+1} \ns{}{+} \xi{}^{2}\ns{}\ns{}A_{k}^{n} \mathbb{V}\mkern-2mu{}w(\ns{}t{\text{\textendash}}n{\text{\textendash}}1\ns{}) A_{k}^{\smash{\T}}{}^{n} \ns{}{+} \mathbb{V}\mkern-2mu{}\tilde{e}_{n}\ns{}(\ns{}t{\text{\textendash}}n\ns{})\ns{}\ns{}\right)\mathclose{}c\\
    &\!\overset{\mathclap{\xi{}\nearrow{}1}}{\leq{}} c^{\smash{\T}}\mathopen{}\left(\ns{}A_{k}^{n+1} \mathbb{V}\mkern-2mu{}e_{0}(t{\text{\textendash}}n{\text{\textendash}}1) A_{k}^{\smash{\T}}{}^{n+1} + A_{k}^{n} \mathbb{V}\mkern-2mu{}w(t{\text{\textendash}}n{\text{\textendash}}1) A_{k}^{\smash{\T}}{}^{n} + \mathbb{V}\mkern-2mu{}\tilde{e}_{n}(t{\text{\textendash}}n)\ns{}\ns{}\right)\mathclose{}c\\
    &\!=c^{\smash{\T}}\mathbb{V}\mkern-2mu{}e_{n+1}(t-n-1)c\,,
  \end{align*}
  since variances are nonnegative and $A_{k}^{n} \mathbb{V}\mkern-2mu{}w(t{\text{\textendash}}n{\text{\textendash}}1) A_{k}^{\smash{\T}}{}^{n}$ is just an additional for the sum in $\tilde{e}$.
  
  The statement follows from applying \cref{eq:proof::lem:cc:closedloop::step} recursively.
\end{proof}
With the recursive feasibility from \cref{thm:recfea}, we can now obtain a qualified statement on chance constraints in closed-loop.

\begin{theorem}\label{thm:cc:closedloop}
  The chance constraints \cref{eq:cc} conditioned on any $x(t_{0})$, where $\xi{}(t_{0}) = 0$, are satisfied in the closed-loop system.
\end{theorem}
\begin{proof}
  Since, for all $t\geq{}t_{0}$, \cref{lem:cc:closedloop} implies 
  \begin{equation*}
    \prob(c_{i}^{\smash{\T}}e(t) \leq{} c_{i}^{\smash{\T}}\mathbb{V}\mkern-2mu{}e_{t-t_{0}}(t_{0})c_{i}\Phi{}^{-1}(\varrho{}_{i})|t_{0}) \geq{} \varrho{}_{i}\,,
  \end{equation*} 
  we apply \cref{thm:recfea,eq:smpc:cc:x} to show that with probability $\varrho{}_{i}$
  \begin{equation*}
    d_{i} \geq{} c_{i}^{\smash{\T}}z(t) + c_{i}^{\smash{\T}}\mathbb{V}\mkern-2mu{}e_{t-t_{0}}(t_{0})c_{i}\Phi{}^{-1}(\varrho{}_{i}) \geq{} c_{i}^{\smash{\T}}z(t) + c_{i}^{\smash{\T}}e(t) \geq{} c_{i}^{\smash{\T}}x(t)
  \end{equation*} 
  holds. The input constraints follow \emph{mutatis mutandis}.
\end{proof}
Since the initialization from \cref{thm:recfea} yields $\xi{}(0)=0$, at the least we achieve the same guarantees as indirect feedback SMPC approaches \cite{Hewing2020indirectSMPC}.
Beyond that the partial feedback on the chance constraints our method incorporates has a two-fold benefit. 

Firstly, the conditioning can be updated, whenever $\xi{}=0$, \eg in case of a disturbance in a beneficial direction. 
Even partial use of the measurement can prevent that the constraint interpretation becomes mostly independent of the current state, when the state has nominally convergent.

Secondly, our initialization enables a smaller constraint tightening by decreasing the variance of the error using the measurement, at least partially.
This secondary benefit is novel compared to previous approaches using initial state interpolation \cite{ist:schluter2022a,Koehler2022interpolSMPC}.

\subsection{Stability and Cost}
As usual the stability guarantee for our scheme derives from the cost function.
Given the choice of the cost function mirrors the standard LQR cost function for SMPC, stability can be shown as follows.
\begin{theorem}
  The SMPC scheme \cref{eq:smpc} is mean square stable for the feasible region, with
  \begin{equation}
    \lim_{r\rightarrow{}\infty{}}\frac1r\sum_{k=0}^r \mathbb{E}{}_{t}\mathopen{}\left[\Vert{}x(t+k)\Vert{}^{2}_Q + \Vert{}u(t+k)\Vert{}^{2}_R\right] \leq{} \ell{}_{ss}\vspace*{0.3\baselineskip}
  \end{equation}
\end{theorem}
\begin{proof}
  The cost function \cref{eq:smpc:cost} is differs from the  standard expected quadratic cost \cref{eq:cost:original} only by a constant. 
  Thus, considering  \cref{eq:cost:original} and following standard arguments, by optimality, we have
  \begin{align*}
      \mathbb{E}{}_t[J^*(t+1)] &\leq{} \mathbb{E}{}_t[\bar J(t+1)]\\
      &\leq{} J^*(t)-(\xVert x(t)\xVert^{2}_Q+\xVert u(t)\xVert^{2}_R-\ell{}_{ss})\,,
  \end{align*}
  where $J^*$ denotes the cost of the optimal solution and $\bar J$ denotes the cost of the candidate from \cref{thm:recfea}. 
  The term $\ell{}_{ss}$ is leftover from the time-shift of the terminal cost.
  Since $J^*$ is finite, taking the expectation conditioned on $x_{0}$ yields the result, for details see \cite[\abvThm7.1]{Kouvaritakis2016book}.
\end{proof}

Unlike in previous methods \cite{Koehler2022interpolSMPC} the interpolation variable $\xi{}$ is already penalized with the normal cost function, since it affects both the variance and the initial error. 
Thus, it is not required to add an artificial cost for $\xi{}$, to numerical reason or otherwise.
Since, in fact, we already arrived at the optimal penalty for allowing an initial deviation between the true system state and nominal state. 

\section{Numerical simulation} 
\begin{figure}\bgroup
  \centering
  \input{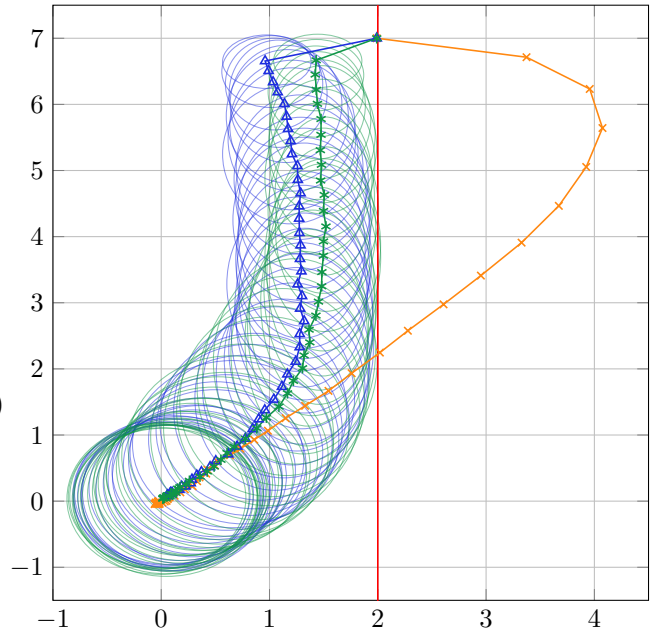}
  \caption{Plot of averaged closed-loop response of iVSMPC in green (\ref{fig:numsim:phaseplot:ivsmpc}) compared to icSMPC (\ref{fig:numsim:phaseplot:icsmpc}) with ellipses indicating an estimate of the 90\%-confidence region. For reference the mean response of LQR (\ref{fig:numsim:phaseplot:lqr}) is shown.}
  \label{fig:numsim:phaseplot}
\egroup\end{figure}
For the numerical study we use the DC-DC-converter regulation problem, which is a widely used benchmark case study in the SMPC literature \cite{Cannon2011stochstictube,Lorenzen2017guaranteedFeasibility} and for this type of schemes in particular \cite{ist:schluter2022a,Ramadan2023outputSMPCic}. 

The corresponding linear dynamics of the form \cref{eq:sys} with
\begin{equation}\label{eq:numsim:example:dcdc}
    A = \begin{bmatrix}1 & 0.0075\\-0.143 & 0.996\end{bmatrix}, \quad
    B = \begin{bmatrix}4.798 \\ 0.115\end{bmatrix}\,,
\end{equation} 
using the increased variance $\Sigma{}_w=0.1\identity_2$ from \cite{ist:schluter2022a,Ramadan2023outputSMPCic}.
The weights for the cost are $Q = \operatorname{diag}\begin{bmatrix}1,&10\end{bmatrix}$ and $R=10$ with a prediction horizon $N=8$.
 
To study the constraint satisfaction we consider a single chance constraint on the first state
\begin{equation}
    \Prob(x^1(k)\leq{}2) \geq{} 0.9
\end{equation}
for the system with an initial state $x(0)=\begin{bmatrix}1.99, & 7\end{bmatrix}^\top{}$ at the border of the initially feasible region, \cf \cref{thm:recfea}.

As direct comparison we consider the initial state interpolation approach to SMPC as in \cite{ist:schluter2022a}, which we will denote as icSMPC in contrast to our iVSMPC.
That approach differs in that the error variance is fixed to $\mathbb{V}\mkern-2mu{}e_\infty{}$, thus the initialization \cref{eq:smpc:ic:V} is missing.
Further the chance constraint tightening differs, relying on Chebyshev's inequality applied to central convex unimodal distribution instead of the CDF of the Gaussian distribution.

Both scheme were implemented with CasADi \cite{Andersson2019casadi} such that the difference is restricted to the outlined points.
Due to the higher dimensionality caused by the need to calculate the variance, our iVSMPC is on average 30\% slower than icSMPC.
This is to be expected, and compensated for by better control performance.
It may be possible to improve the computation time by precomputing the variance trajectories and then only applying a scaling online.

\subsection{Constraint satisfaction in closed-loop} 
In order to illustrated constraint satisfaction, we rely on the dynamics, which violate the constraint under unconstrained optimal control. 
Thereby, we can show that the scheme follows the constraint border to stay as close as possible to the optimal trajectory.
Therefore, the chance constraint is active for a considerable amount of steps, before converging.

In \cref{fig:numsim:phaseplot}, we contrast iVSMPC with the icSMPC and LQR trajectories as obtained from a Monte Carlo simulation over 1\,100\,000 different disturbance realizations.
First of all, we observe that the trajectories under the SMPC schemes smoothly slide along the tightened constraint.
The difference of cause being that the variable-sized tube of iVSMPC allow a closer approach to the constraint.

Further, by explicitly considering the variance in iVSMPC, the scheme can actively reduce the close-loop variance as necessary, which allows to further approach the constraint. 

Lastly, unlike icSMPC the constraint tightening of the iVSMPC is exact, thus the confidence bounds actually touch the constraint.
Where the conservative Chebyshev tightening of icSMPC, hinders a closer tracking of the constraint.

\begin{figure}
  \begin{tikzpicture}
    \begin{axis}[
        grid=major,
        xmin=0, xmax=3, ymin=0, ymax=1,
        xtick distance=0.5, x dir=reverse,
        ytick distance=0.1,
        height=0.5\linewidth, width=0.9\linewidth,
        colormap={RdYlBu_r}{rgb(0.0)=(0.19215686274509805,0.21176470588235294,0.58431372549019611) rgb(0.1)=(0.27058823529411763,0.45882352941176469,0.70588235294117652) rgb(0.2)=(0.45490196078431372,0.67843137254901964,0.81960784313725488) rgb(0.3)=(0.6705882352941176,0.85098039215686272,0.9137254901960784) rgb(0.4)=(0.8784313725490196,0.95294117647058818,0.97254901960784312) rgb(0.5)=(1.0,1.0,0.74901960784313726) rgb(0.6)=(0.99607843137254903,0.8784313725490196,0.56470588235294117) rgb(0.7)=(0.99215686274509807,0.68235294117647061,0.38039215686274508) rgb(0.8)=(0.95686274509803926,0.42745098039215684,0.2627450980392157) rgb(0.9)=(0.84313725490196079,0.18823529411764706,0.15294117647058825) rgb(1.0)=(0.6470588235294118,0.0,0.14901960784313725)},
        colorbar,
        ]
        \fill[red,opacity=0.1] (axis cs:3,1) rectangle (axis cs:2,0.1);
        \addplot[point meta min=0,point meta max=1] graphics[xmin=0,ymin=0,xmax=3,ymax=1] {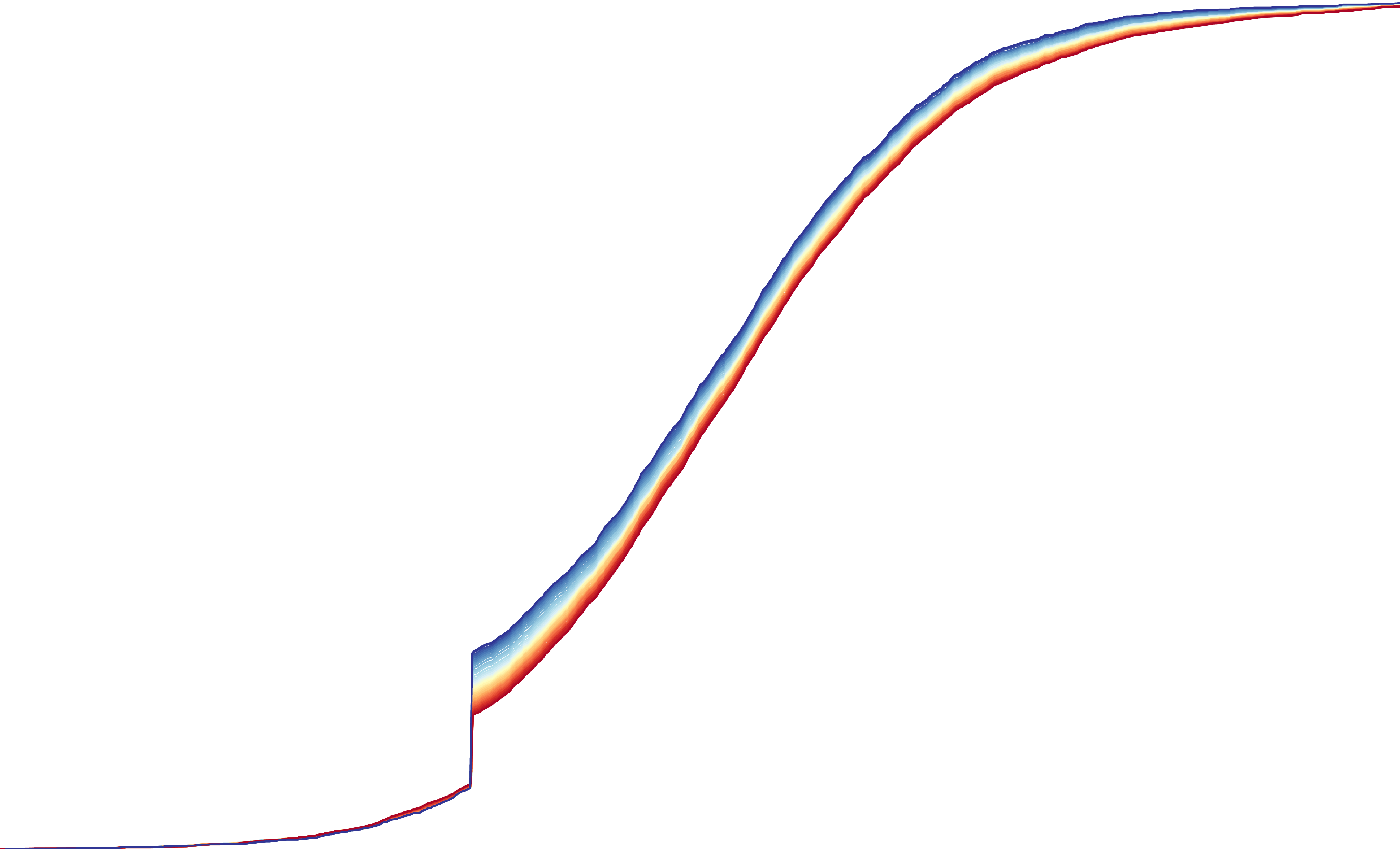};
    \end{axis}
\end{tikzpicture}
  \caption{Empirical distribution functions of $\prob(x(t+1)|\xi{}(t)<\bar \xi{},t<25)$ for different $\bar \xi{}\in{}[0,1]$, while the chance constraint is active, \ie within the first 25 steps.
  The red region indicates the chance constraint.}
  \label{fig:numsim:ecdf}
\end{figure}
\Cref{fig:numsim:ecdf} shows the empirical distribution function (eCDF) for the state at the next time-step conditioned on $\xi{}$ smaller than some level $\bar \xi{}$ and that the chance constraint is active.
As seen in \cref{fig:numsim:phaseplot} the constraint is active for the first 25 time-steps.
This figure indicates not only that the chance constraint is satisfied for $\xi{}=0$ (top-most eCDF), as guaranteed by \cref{thm:cc:closedloop}, but additional we observe in this example that the chance constraint can also be reconditioned with a $\xi{}>0$.
This would allow for a cleaner interpretation of the chance constraint in close-loop, since we can usually recondition on the current measurement.
Thus, the delay and, thereby, the uncertainty increase over time is less relevant.

\subsection{Closed-loop performance} 
Not only are we interested in tighter constraint satisfaction, but iVSMPC also promises an improved performance as a consequence thereof.
This can also be observed in the Monte-Carlo simulation.
Of course, the LQR cost cannot be achieved under the constraint, yet since the optimal constraint cost is unknown, it serves as a good reference point.
Thus, normalizing the average cost with respect to the LQR cost, we have that the cost for the iVSMPC  is $178\%$, whereas for the icSMPC $207\%$ of the LQR cost.
This is the immediate benefit of risking a trajectory closer to the chance constraint. 

\subsection{Interpolating variable} 
\normalcolor
\begin{figure}
  \bgroup
  \begin{tikzpicture}
    \begin{axis}[
        name=ivsmpc,
        grid=none,
        xmin=0, xmax=50, ymin=0, ymax=1,
        xtick distance=10, 
        ytick distance=1,
        height=0.35\linewidth, width=1.1\linewidth,
    ]
        \addplot graphics[xmin=0, xmax=50, ymin=0, ymax=1] {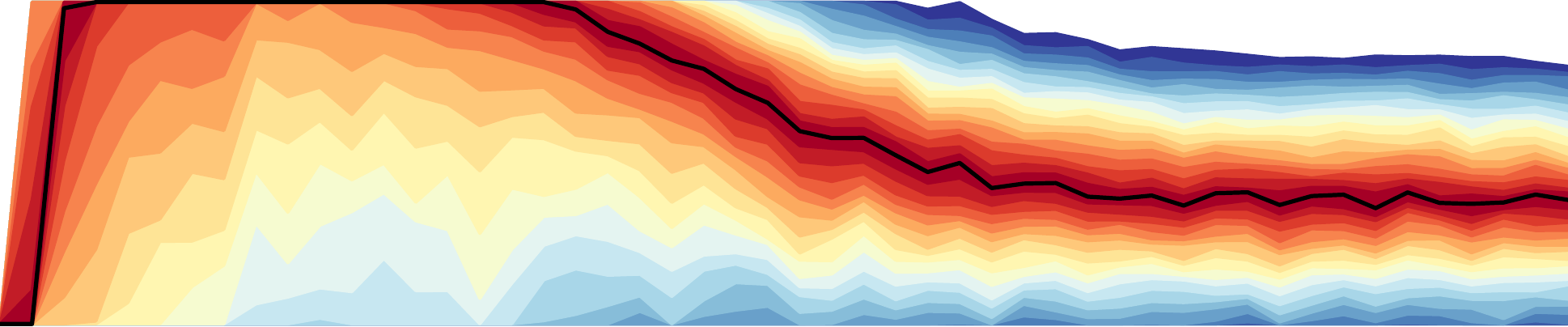};
    \end{axis}
    \begin{axis}[
        name=icsmpc,
        at={(ivsmpc.below south west)},anchor={north west},
        grid=none,
        xmin=0, xmax=50, ymin=0, ymax=1,
        xtick distance=10, 
        ytick distance=1,
        height=0.35\linewidth, width=1.1\linewidth,
        xticklabels={},
    ]
        \addplot graphics[xmin=0, xmax=50, ymin=0, ymax=1] {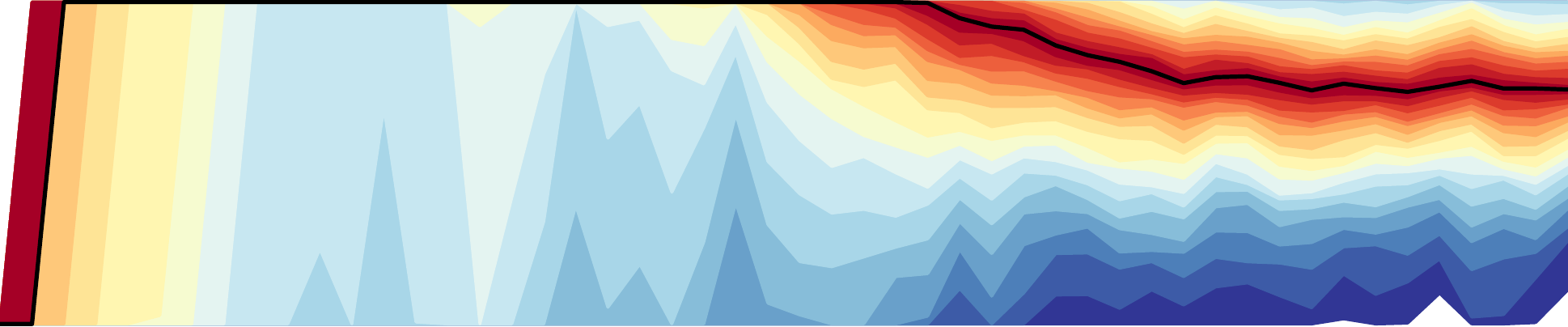};
    \end{axis}
\end{tikzpicture}
\vspace*{-1.1\baselineskip}
  \egroup
  \caption{The empirical distribution of $\xi{}$ for iVSMPC (top) and icSMPC (bottom) with the median in black. }
  \label{fig:numsim:xi}
\end{figure}
From the derivation of the cost function, \cf \cref{sec:smpc:cost}, we have obtained a scheme that has an explainable cost to the choice of the interpolation variable $\xi{}$.
The older icSMPC has no such insight available.
In \cref{fig:numsim:xi}, we study the impact thereof.
We observe that iVSMPC makes more use of the actual measurement than icSMPC. 
In particular, towards the end, \ie while approaching the steady state, we see the iVSMPC rarely ignores the measurement $\xi{}\to 1$. 
The icSMPC, however, has not clear preference, but tends towards ignoring the measurement.
This obviously reduces nominal cost, since the nominal state converges.
However, this risks higher cost for the true state $x$ as the error increases.

\section{Conclusion} 
We have presented an SMPC scheme for linear system subject to additive Gaussian disturbances, guaranteeing closed-loop constraint satisfaction and stability.
By initializing both the nominal state and the error variance, based on the optimized interpolation variable, we improve upon previous approaches \cite{ist:schluter2022a,Koehler2022interpolSMPC,Pan2022ddOutputSmpcIC,Mark2023dddrsmpcic}, enabling improved control performance with variable-sized tubes.
Using a simple example we illustrate the better performance compared to previous methods, which is mainly due to operation closer to the constraints becoming feasible with the variable-sized tubes adapting to the situation.
Beyond that we highlight numerically the flexibility of the proposed approach, allowing a flexible interpretation of the chance constraint, and using more of the measurement than previous methods.
Future work will include developing a more rigorous understanding of the case where the chance constraint needs not be conditioned on exactly zero error.
Beyond that extensions to more general system classes are worthwhile to explore.

\printbibliography
\end{document}